\documentclass[12pt,UKenglish]{article}

\usepackage{babel,csquotes}
\usepackage{amssymb}
\usepackage{amsmath,amsthm}
\usepackage{amsfonts}
\usepackage{setspace}
\usepackage{rotating}
\usepackage{bm}
\usepackage{graphicx,psfrag,epsf}
\usepackage{enumerate}
\usepackage[style=authoryear]{biblatex} 
\usepackage{url}
\usepackage{dcolumn}
\usepackage{epstopdf}
\usepackage{multirow}
\usepackage{booktabs}

\input{ee.sty}

\newcommand{\blind}{0}

\newcolumntype{.}{D{.}{.}{-1}}
\newtheorem{theorem}{Theorem}
\newtheorem{lemma}{Lemma}

\begin{document}

\def\spacingset#1{\renewcommand{\baselinestretch}%
{#1}\small\normalsize} \spacingset{1}

\if0\blind
{
  \title{\bf Can $GDP$ measurement be further improved?\\Data revision and reconciliation}
  \author{Jan P.A.M. Jacobs\thanks{Preliminary versions of this paper were presented at the 10th International Conference on Computational and Financial Econometrics (CFE 2016), the XIII Conference on Real-Time Data Analysis, Methods and Applications, Banco de Espa\~na, and the ESCoE Conference on Economic Measurement 2018, London. We thank Dean Croushore, Gabriele Fiorentini, Adrian Pagan, Alexander Rathke and Enrique Sentana for helpful comments.}\hspace{.2cm}\\
    \small{University of Groningen, University of Tasmania, CAMA and CIRANO}~
\and Samad Sarferaz 
\\ \small{KOF Swiss Economic Institute, ETH Zurich, Switzerland}
\and Jan-Egbert Sturm \\ \small{KOF Swiss Economic Institute, ETH Zurich, Switzerland and CESifo, Germany} 
\and Simon van Norden \\ \small{HEC Montr\'eal, CIRANO and CIREQ}
    }
\date{August 2018} 
  \maketitle
} \fi

\if1\blind
{
  \bigskip
  \bigskip
  \bigskip
  \begin{center}
    {\LARGE\bf Can $GDP$ measurement be further improved? \\Data revision and reconciliation}
\end{center}
  \medskip
\date{August 2018} 
} \fi

\bigskip
\begin{abstract}
\noindent Recent years have seen many attempts to combine expenditure-side estimates of U.S. real output ($GDE$) growth with income-side estimates ($GDI$) to improve estimates of real GDP growth. We show how to incorporate information from multiple releases of noisy data to provide more precise estimates while avoiding some of the identifying assumptions required in earlier work. This relies on a new insight: using multiple data releases allows us to distinguish news and noise measurement errors in situations where a single vintage does not.

\noindent Our new measure, $GDP^{++}$, fits the data better than $GDP^+$, the $GDP$ growth  measure of Aruoba et al. (2016)\nocite{Aruobaetal2016} published by the Federal Reserve Bank of Philadephia. Historical decompositions show that $GDE$ releases are more informative than $GDI$, while the use of multiple data releases is particularly important in the quarters leading up to the Great Recession. 
\end{abstract}

\noindent%
\\[0.1in]
\vfill

\newpage
\setcounter{page}{1}
\spacingset{1.45} 

\section{Introduction}
Unlike many other nations, U.S. national accounts feature distinct estimates of real output based on the expenditure approach ($GDE$) and the income approach ($GDI$), see Figure \ref{GDP-GDI}. As pointed out by Stone, Champernowne and Meade (1942)\nocite{StoneChampernowneMeade1942}, while in theory these two approaches should give identical estimates, measurement errors cause discrepancies to arise.\footnote{%
The same applies to the production-based estimate of output. See e.g. the study of Rees, Lancaster and Finlay (2015)\nocite{ReesLancasterFinlay2015} on Australian GDP.} 
These discrepancies are sometimes important. Chang and Li (2015)\nocite{ChangLi2015} examine the impact of using $GDI$ rather than $GDE$ in nearly two dozen recent empirical papers published in major economic journals; they find substantive differences in roughly 15\% of them. Nalewaik (2012)\nocite{Nalewaik2012} finds that $GDI$ leads to quicker detection of U.S. recessions than $GDE$.

\begin{figure}[ht] 
\caption{U.S. $GDP$ growth: Expenditure side vs. income side\label{GDP-GDI}}
\centering
    \includegraphics[width=\textwidth]{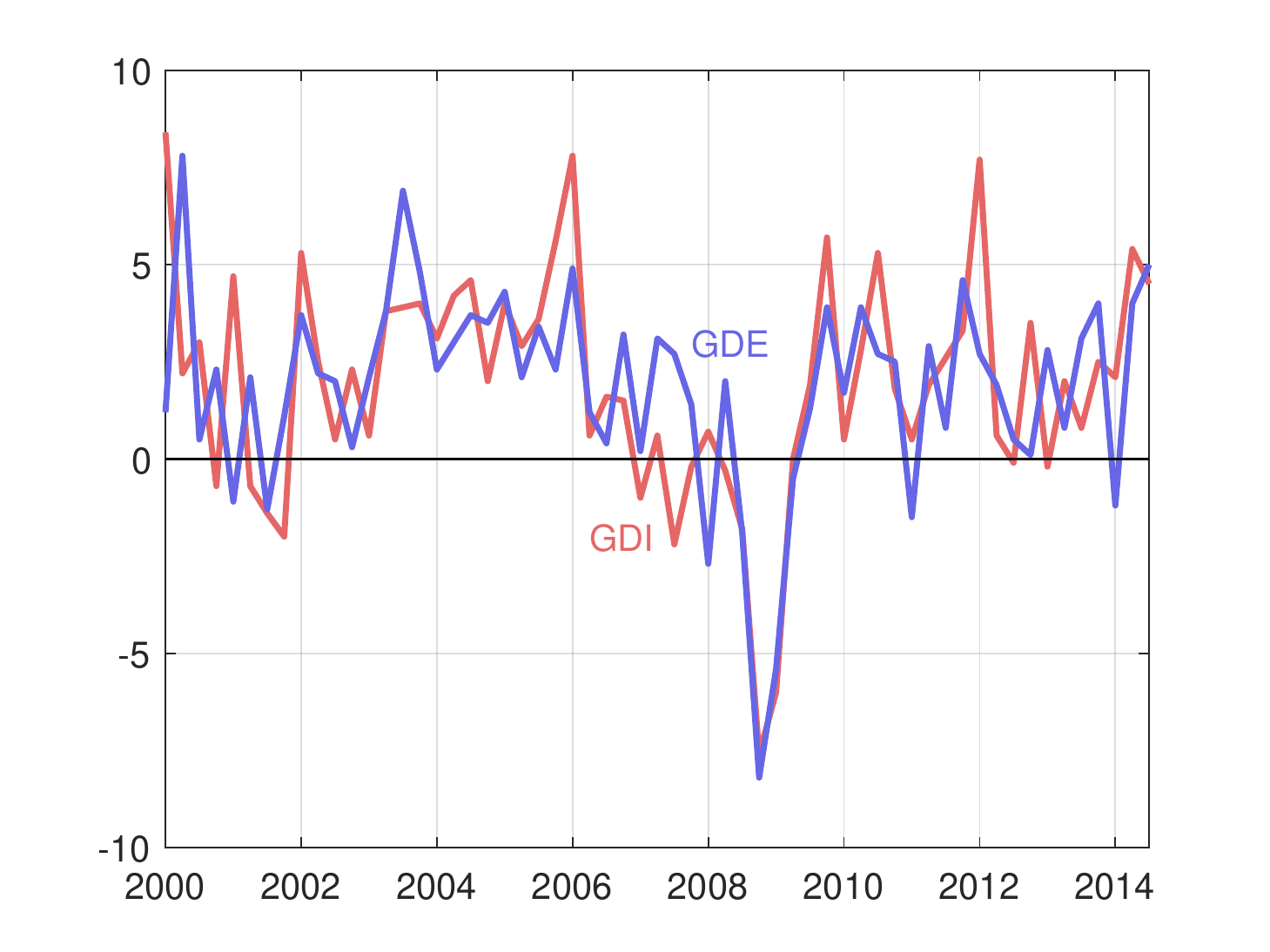}
\end{figure}

\newpage
While several studies have tried to determine which measure should be preferred in various contexts, Weale (1992)\nocite{Weale1992} and Diebold (2010)\nocite{Diebold2010} argue that reconciling them is a more useful response as it should incorporate more information. Fixler and Nalewaik (2009)\nocite{FixlerNalewaik2009} point out, however, that reconciliation traditionally relies on the assumption that measurement errors are ``noise'', which in turn forces the reconciled estimate of the latent variable (``true'' $GDP$ in this case) to be less variable than any of the individual series being reconciled. They instead propose that measurement errors may also include a ``news'' component. While this causes a loss of identification, they glean information from the revision of $GDE$ and $GDI$ to place bounds on relative contributions of news and noise in a least-squares framework. Aruoba et al. (2012)\nocite{Aruobaetal2012} consider the problem from a forecast combination perspective, assuming ``news'' errors and imposing priors in lieu of identification without revisions, while Aruoba et al. (2016)\nocite{Aruobaetal2016} consider alternative identifying assumptions and propose the addition of an instrumental variable. Almuzara et al. (2018)\nocite{AlmuzaraFiorentiniSentana2018} investigate a dynamic factor model with cointegration restrictions. 


\begin{figure}[ht]
\caption{$GDP^+$ in real-time\label{GDPplus_revisions}}
\begin{center}
\includegraphics[width=\textwidth]{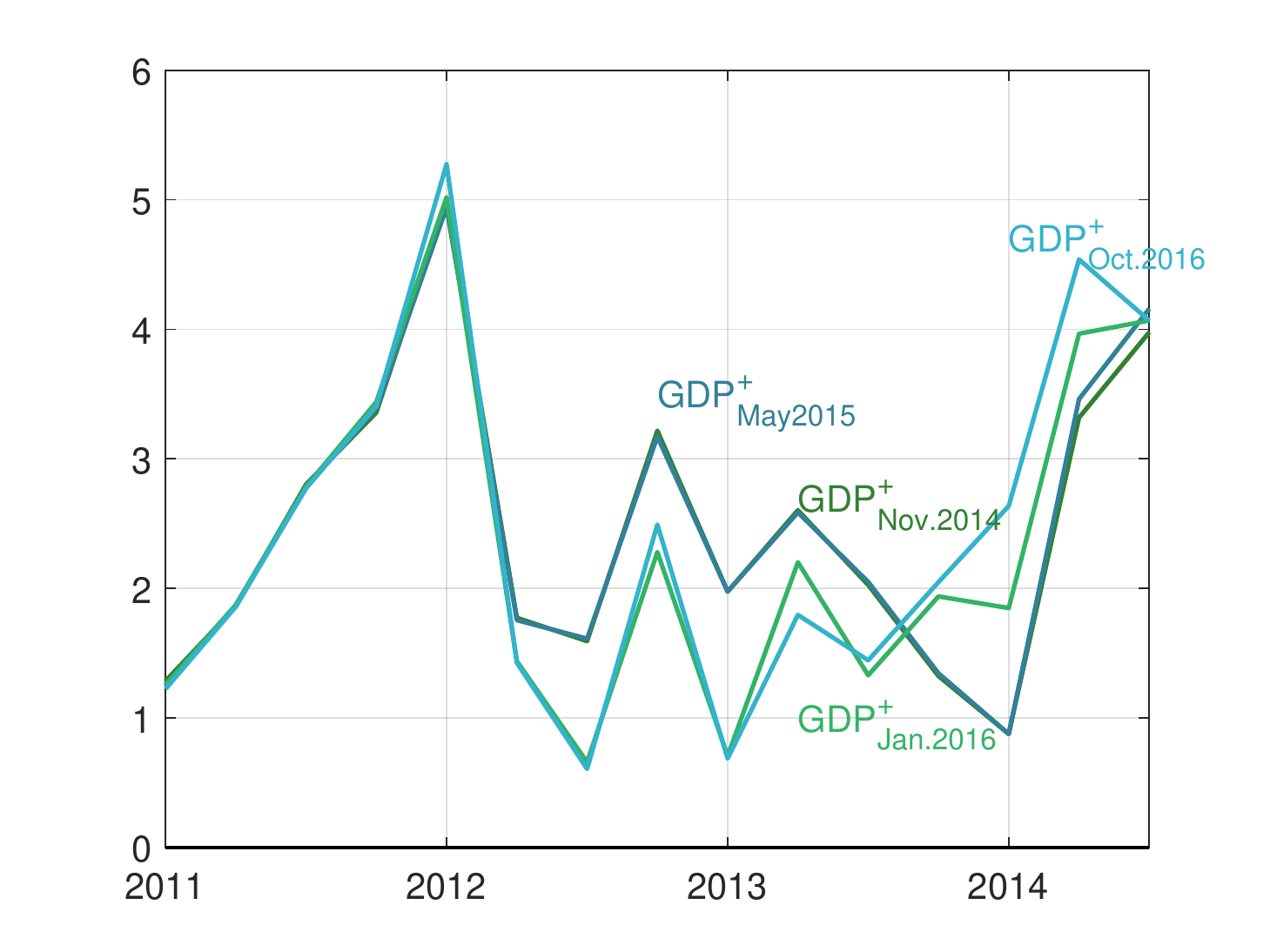}
\end{center}
{\footnotesize Various vintages of $GDP^+$ according to the estimates of the Federal Reserve Bank of Philadelphia.}
\end{figure}

Aruoba et al. (2016)\nocite{Aruobaetal2016} is the basis for the $GDP^+$ measure published by the Federal Reserve Bank of Philadelphia.\footnote{See \texttt{http://www.philadelphiafed.org/research-and-data/real-time-center/gdpplus/}} However, while their approach ignores the possibility of data revision, Figure \ref{GDPplus_revisions} shows that the published series is subject to important revisions, which complicates its interpretation and use in policy decisions.
Separately, Jacobs and van Norden (2011)\nocite{JacobsvanNorden2011} and Kishor and Koenig (2012)\nocite{KishorKoenig2012} propose state-space frameworks that allow estimation of both news- and noise-type measurement errors in data revision, but do not consider problems of data reconciliation. In this paper we extend Jacobs and van Norden (2011\nocite{JacobsvanNorden2011}, henceforth JvN) to consider the problem of reconciliation and identification in which there are multiple estimates of the common underlying variable, all of which are subject to revision. Allowing for both news and noise measurement errors, the result is a modeling framework substantially more general than those previously proposed. We show that identification of these two types of measurement errors is made possible by modeling data revisions as well as the dynamics of the series. We provide a historical decomposition of $GDE$ and $GDI$ into news and noise shocks, and we compare those series to our improved $GDP$ estimate, $GDP^{++}$. We find that $GDP^{++}$ is more persistent than either $GDE$ or $GDI$. While both series appear to contain both news and noise shocks, news shocks have a larger share in $GDE$ than in $GDI$.

The paper is structured as follows. In Section \ref{framework} we present our econometric framework. We show that our system is identified using real-time data and news-noise assumptions. In Section \ref{data_estimation} we describe our data and estimation method. Results are shown in Section \ref{results} and Section \ref{conclusion} concludes. Formal proofs of some results related to identification are presented in an Appendix.

\section{Econometric Framework\label{framework}}


In this section, after establishing some notation, we describe our
econometric framework. We begin by briefly reviewing the univariate news and noise model of
JvN before generalizing it to the problem of data reconciliation. We then
compare the results to the $GDP^{+}$ model of Aruoba et al. (2016) and
discuss their differences for the identification of news and noise
measurement errors. 

We follow the standard notation in this literature by letting $y_{t}^{t+j}$
be an estimate published at time $t+j$ of some real-valued scalar
variable $y$ at time $t$. We define $\vy_{t}$ as a $l\times 1$ vector of $l$
different vintage estimates of $y_{t}^{t+i}$, $i=1,\ldots ,l$ so $\vy%
_{t}\equiv \left[ y_{t}^{t+1},y_{t}^{t+2},\ldots ,y_{t}^{t+l}\right]
^{\prime }$. \ For state-space models, we follow the notation of Durbin and
Koopman (2001)\nocite{DurbinKoopman2001} 
\begin{align}
\vy_{t}& =\mZ\cdot \valpha_{t}+\vvarepsilon_{t}  \label{measurementequation}
\\
\valpha_{t+1}& =\mT\cdot \valpha_{t}+\mR\cdot \veta_{t}
\label{transitionequation}
\end{align}%
where $\vy_{t}$ is $l\times 1$, $\valpha_{t}$ is $m\times 1$, $\vvarepsilon%
_{t}$ is $l\times 1$ and $\veta_{t}$ is $r\times 1$; $\vvarepsilon_{t}\sim
N(0,\mH)$ and $\veta_{t}\sim $ $N(0,\mI_{r})$. Both error terms are i.i.d.\
and orthogonal to one another.\footnote{%
For more detailed assumptions, see Durbin and Koopman (2001\nocite%
{DurbinKoopman2001}, Section 3.1 and 4.1. For convenience we omit constants
from the model in this exposition. }

\subsection{A State-Space model of Measurement Error with News and Noise}

JvN denote the unobserved \textquotedblleft true\textquotedblright\ value of
a variable as $\tilde{y}_{t}$, so that its measurement error $\vu_{t}\equiv %
\vy_{t}-\viota_{l}\cdot \tilde{y}_{t}$, where $\viota_{l}$ is an $l\times 1$
vector of ones. They model these measurement errors as the sum of ``news''
and ``noise'' measurement errors. Measurement errors are said to be noise $%
\left( \zeta _{t}^{t+i}\right) $ when they are orthogonal to the true values 
$\tilde{y}_{t}$, so that 
\begin{equation}
y_{t}^{t+i}=\tilde{y}_{t}+\zeta _{t}^{t+i},\qquad \cov(\tilde{y}_{t},\zeta
_{t}^{t+i})=0.  \label{noisemodel}
\end{equation}%
Noise implies that revisions ($y_{t}^{t+i+1}-y_{t}^{t+i}$) are generally
forecastable. Measurement errors are described as news $(\nu _{t}^{t+i})$ if
and only if  
\begin{equation}
\tilde{y}_{t}=y_{t}^{t+i}+\nu _{t}^{t+i},\qquad \cov(y_{t}^{t+j},\nu
_{t}^{t+i})=0\qquad \forall j\leq i  \label{newsmodel}
\end{equation}%
If data revisions are pure news errors, current and past vintages of the
series will be of no use in forecasting future data revision. 

In their state-space model JvN impose $\vvarepsilon_{t}\equiv \bm{0}_{l\times 1}$
and partition the state vector $\valpha_{t}$ into four components 
\begin{equation}
\valpha_{t}=\left[ \tilde{y}_{t},\vphi_{t}^{\prime },\vnu_{t}^{\prime },%
\vzeta_{t}^{\prime }\right] ^{\prime },  \label{partitioned_alpha}
\end{equation}%
of length $1,$ $b$, $l$ and $l$ respectively, where $\vphi_{t}$ is used to
capture the dynamics of the true values while $\vnu_{t}$ and $\vzeta_{t}$
are the news and noise measurement errors, respectively. They similarly
partition 
\begin{equation}
\mZ=\left[ \mZ_{1},\mZ_{2},\mZ_{3},\mZ_{4}\right]   \label{partitioned_Z}
\end{equation}%
where $\mZ_{1}=\viota_{l}$ (a $l\times 1$ vector of 1's), $\mZ_{2}=\bm{0}%
_{l\times b}$ (an $l\times b$ matrix of zeros), $\mZ_{3}=\mI_{l}$, and $\mZ%
_{4}=\mI_{l}$ (both $l\times l$ identity matrices). Their measurement
equation (\ref{measurementequation}) then simplifies to 
\begin{equation}
\vy_{t}=\mZ\cdot \valpha_{t}=\tilde{y}_{t}+\vnu_{t}+\vzeta_{t}=\mbox{`Truth'}%
+\mbox{`News'}+\mbox{`Noise'}.  \label{eq:JvN_Measurement}
\end{equation}%
They conformably partition the matrix $\mT$ as 
\begin{equation}
\mT=%
\begin{bmatrix}
T_{11} & \mT_{12} & \bm{0} & \bm{0} \\ 
\mT_{21} & \mT_{22} & \bm{0} & \bm{0} \\ 
\bm{0} & \bm{0} & \mT_{3} & \bm{0} \\ 
\bm{0} & \bm{0} & \bm{0} & \mT_{4}%
\end{bmatrix}%
,  \label{partitioned_T}
\end{equation}%
where $T_{11}$ is a scalar, and $\left\{ \mT_{12},\mT_{21},\mT_{22},\mT_{3},%
\mT_{4}\right\} $ are $1\times b$, $b\times 1$, $b\times b$, $l\times l$ and 
$l\times l$; $\bm{0}$ is a conformably defined matrix of zeros. The $\left(
b+1\right) \times \left( b+1\right) $ block in the upper left simply
captures the dynamics of $\widetilde{y}_{t}$ while $\mT_{3}$ and $\mT_{4}$
capture the dynamics of the news and noise shocks. If measurement errors are
independent across time periods (but not vintages), then $\mT_{3}\equiv \mT%
_{4}\equiv \bm{0}_{l\times l}$. \ As we will see below, in the special case where $%
\tilde{y}_{t}$ is assumed to follow an $AR\left( p\right) $ process, this
will impose $p=b+1$, the row vector $%
\begin{bmatrix}
T_{11} & \mT_{12}%
\end{bmatrix}%
$ will contain the autoregressive coefficients and the remainder of the
upper left $\left( b+1\right) \times \left( b+1\right) $ part will be composed of
zeros and ones.\footnote{%
For details, see Jacobs and van Norden (2011)\nocite{JacobsvanNorden2011}.}

The essential difference between news and noise errors is captured in 
the $\left( 1+b+2l\right) \times (1+2l)$ matrix $\mR,$
which is partitioned as follows
\begin{equation}
\mR=%
\begin{bmatrix}
\mR_{1} & \mR_{3} & \bm{0} \\ 
\mR_{2} & \bm{0} & \bm{0} \\ 
\bm{0} & -\mU_{l}\cdot \diag(\mR_{3}) & \bm{0} \\ 
\bm{0} & \bm{0} & \mR_{4}%
\end{bmatrix}%
,  \label{partitioned_R}
\end{equation}%
where $\mU_{l}$ is a $l\times l$ matrix with zeros below the main diagonal
and ones everywhere else, $\mR_{3}=\left[ \sigma _{\nu 1},\sigma _{\nu
2},\ldots ,\sigma _{\nu l}\right] $, where $\sigma _{\nu i}$ is the standard
error of the measurement error associated with $i$-th estimate $y_{t}^{t+i}$%
, $\diag(\mR_{3})$ is a $l\times l$ matrix with elements of $\mR_{3}$ on its
main diagonal, and $\mR_{4}$ is an $l\times l$ matrix. Finally, the error
term is partitioned as $\veta_{t}=\left[ \veta_{et}^{\prime },\ \veta_{\nu
t}^{\prime },\ \veta_{\zeta t}^{\prime }\right] ^{\prime }$, where $\veta%
_{et}$ refers to errors associated with the true values, and $\veta_{\nu t}$
and $\veta_{\zeta t}$ are the errors for news and noise, respectively.

JvN note that (if the model is identified, a question we deal with below)
this framework permits conventional techniques to be used to
estimate the model parameters, allow for missing observations, estimate and
forecast the unobserved true values $\widetilde{y}_{t}$ together with their
confidence intervals, and test hypotheses. 

\subsection{Data Reconciliation}

We now show how the above framework may be adapted to the case where we
have two alternative estimates of the same underlying true value $\widetilde{%
y}_{t}$, both of which are subject to revision. We define $\mY_{t}$ as a $%
2l\times 1$ vector of $l$ different vintage estimates for the $2$ variables $%
y1_{t}^{t+i}$ and $y2_{t}^{t+i}$, $i=1,\ldots ,l$, for a particular
observation $t$, so $\mY_{t}\equiv \left[ y1_{t}^{t+1},y1_{t}^{t+2}\ldots
,y1_{t}^{t+l},y2_{t}^{t+1},y2_{t}^{t+2},\ldots ,y2_{t}^{t+l}\right] ^{\prime
},$ a vector of length $2l.$\ Our state-space model now becomes 
\begin{align}
\mY_{t}& =\mZ\cdot \valpha_{t}  \label{eq:Measurement_new} \\
\valpha_{t+1}& =\mT\cdot \valpha_{t}+\mR\cdot \veta_{t}
\label{eq:Transition_new}
\end{align}%
We again partition the state vector $\valpha_{t}$ into four components 
\begin{equation}
\valpha_{t}=\left[ \tilde{y}_{t}^{\prime },\vphi_{t}^{\prime },\vnu%
_{t}^{\prime },\vzeta_{t}^{\prime }\right] ^{\prime },
\label{eq:alphaPartition}
\end{equation}%
which are now of length $1,$ $b$, $2l$ and $2l$ respectively, and we
similarly partition 
\begin{equation}
\mZ=\left[ \mZ_{1},\mZ_{2},\mZ_{3},\mZ_{4}\right]   \label{eq:Zpartition}
\end{equation}%
where $\mZ_{1}=\viota_{2l}$(a $2l$ vector of ones), $\mZ_{2}=\bm{0}%
_{2l\times b}$ (a $2l\times b$ matrix of zeros), and $\mZ_{3}=$ $\mZ_{4}=\mI%
_{2l}$ (both are $2l\times 2l$ identity matrices). The measurement equation (%
\ref{eq:Measurement_new}) therefore again simplifies to 
\[
\mY_{t}=\mZ\cdot \valpha_{t}=\tilde{y}_{t}+\vnu_{t}+\vzeta_{t}=\mbox{`Truth'}%
+\mbox{`News'}+\mbox{`Noise'}.
\]%
The matrix $\mT$ is partitioned much as before%
\begin{equation}
\mT=%
\begin{bmatrix}
T_{11} & \mT_{12} & \bm{0} & \bm{0} \\ 
\mT_{21} & \mT_{22} & \bm{0} & \bm{0} \\ 
\bm{0} & \bm{0} & \mT_{3} & \bm{0} \\ 
\bm{0} & \bm{0} & \bm{0} & \mT_{4}%
\end{bmatrix}%
,  \label{eq:Tpartition}
\end{equation}%
The upper left block (consisting of $T_{11},\mT_{12},\mT_{21}$ and $\mT_{22}$%
) is precisely the same as in (\ref{partitioned_T}) above; this is because
it solely determines the dynamics of $\widetilde{y}_{t}$, which are
unchanged. However, the addition of a new series increases the dimension of $%
\mT_{3}$ and $\mT_{4}$ from $l\times l$ to $2l\times 2l$. 

$\ \mR$ is now a $%
\left( 1+b+4l\right) \times (1+4l)$ matrix where we separate the news and
noise measurement errors for the two variables 
\begin{equation}
\mR=%
\begin{bmatrix}
\mR_{1} & \mR_{3} & \mR_{4} & \bm{0} & \bm{0} \\ 
\mR_{2} & \bm{0} & \bm{0} & \bm{0} & \bm{0} \\ 
\bm{0} & -\bm{U}_{l}\cdot \diag(\mR_{3}) & \bm{0} & \bm{0} & \bm{0} \\ 
\bm{0} & \bm{0} & -\bm{U}_{l}\cdot \diag(\mR_{4}) & \bm{0} & \bm{0} \\ 
\bm{0} & \bm{0} & \bm{0} & \mR_{5} & \bm{0} \\ 
\bm{0} & \bm{0} & \bm{0} & \bm{0} & \mR_{6}%
\end{bmatrix}%
,  \label{eq:Rpartitioned}
\end{equation}%
where the row vector $\mR_{3}=\left[ \sigma _{\nu _{1}^{1}},\sigma _{\nu
_{2}^{1}},\ldots ,\sigma _{\nu _{1}^{1}}\right] $ corresponds to the news in 
$y1$ while $\mR_{4}=\left[ \sigma _{\nu _{1}^{2}},\sigma _{\nu
_{2}^{2}},\ldots ,\sigma _{\nu _{l}^{2}}\right] $ corresponds to the news in 
$y2$. $\diag(\mR_{3})$ and $\diag(\mR_{4})$ are $l\times l$ diagonal
matrices with the elements of $\mR_{3}$ and $\mR_{4}$ on their main
diagonals, while $\mR_{5}$ and $\mR_{6}$ are $l\times l$ diagonal matrices.

Finally, we partition $\veta_{t}=\left[ \veta_{et}^{\prime },\ \veta_{\nu
_{1}t}^{\prime },\ \veta_{\nu _{2}t}^{\prime },\ \veta_{\zeta _{1}t}^{\prime
},\ \veta_{\zeta _{2}t}^{\prime }\right] ^{\prime }$, where $\veta_{et}$
refers to errors associated with the true values, and $\veta_{\nu it}$ and $%
\veta_{\zeta it}$ are the errors for news and noise measurement errors in
variable $i.$

\medskip
To illustrate, consider the following very simple case. Let $y1\equiv GDE$
(the growth rate of real gross domestic expenditure), $y2\equiv GDI$ (the
growth rate of real gross domestic income), $l=2$ (we only consider two
vintages, the 1st and 2nd releases) and we'll assume that the growth rate of
``true" real output $\widetilde{y}$ follows an $AR\left( 1\right) $. Then (%
\ref{eq:Measurement_new}) becomes 
\begin{eqnarray*}
\begin{bmatrix}
GDE^{1st}_{t} \\ 
GDE^{2nd}_{t} \\ 
GDI^{1st}_{t} \\ 
GDI^{2nd}_{t}%
\end{bmatrix}%
 &=&%
\begin{bmatrix}
1 & 0 & 1 & 0 & 0 & 0 & 1 & 0 & 0 & 0 \\ 
1 & 0 & 0 & 1 & 0 & 0 & 0 & 1 & 0 & 0 \\ 
1 & 0 & 0 & 0 & 1 & 0 & 0 & 0 & 1 & 0 \\ 
1 & 0 & 0 & 0 & 0 & 1 & 0 & 0 & 0 & 1%
\end{bmatrix}%
\cdot 
\begin{bmatrix}
\tilde{y}_{t} \\ 
\tilde{y}_{t-1} \\ 
\vnu_{t} \\ 
\vzeta_{t}%
\end{bmatrix}
\\
&=&%
\begin{bmatrix}
\tilde{y}_{t} \\ 
\tilde{y}_{t} \\ 
\tilde{y}_{t} \\ 
\tilde{y}_{t}%
\end{bmatrix}%
+%
\begin{bmatrix}
\nu _{t}^{GDE,1} & 0 & 0 & 0 \\ 
0 & \nu _{t}^{GDE,2} & 0 & 0 \\ 
0 & 0 & \nu _{t}^{GDI,1} & 0 \\ 
0 & 0 & 0 & \nu _{t}^{GDI,2}%
\end{bmatrix}%
+%
\begin{bmatrix}
\zeta _{t}^{GDE,1} & 0 & 0 & 0 \\ 
0 & \zeta _{t}^{GDE,2} & 0 & 0 \\ 
0 & 0 & \zeta _{t}^{GDI,1} & 0 \\ 
0 & 0 & 0 & \zeta _{t}^{GDI,2}%
\end{bmatrix}
\\
&=&\mbox{`Truth'}+\mbox{`News'}+\mbox{`Noise'}.
\end{eqnarray*}
and (\ref{eq:Transition_new}) becomes%
\[
\begin{bmatrix}
\tilde{y}_{t+1} \\ 
\tilde{y}_{t} \\ 
\vnu_{t+1} \\ 
\vzeta_{t+1}%
\end{bmatrix}%
=%
\begin{bmatrix}
\rho  & 0 & \bm{0} & \bm{0} \\ 
1 & 0 & \bm{0} & \bm{0} \\ 
\bm{0} & \bm{0} & \mT_{3} & \bm{0} \\ 
\bm{0} & \bm{0} & \bm{0} & \mT_{4}%
\end{bmatrix}%
\cdot 
\begin{bmatrix}
\tilde{y}_{t} \\ 
\tilde{y}_{t-1} \\ 
\vnu_{t} \\ 
\vzeta_{t}%
\end{bmatrix}%
+\mR\cdot \veta_{t},
\]%
where 
\begin{eqnarray*}
\mR &=&%
\begin{bmatrix}
\sigma _{e} & \sigma _{\nu }^{GDE1} & \sigma _{\nu }^{GDE2} & \sigma _{\nu
}^{GDI1} & \sigma _{\nu }^{GDI2} & 0 & 0 & 0 & 0 \\ 
0 & 0 & 0 & 0 & 0 & 0 & 0 & 0 & 0 \\ 
0 & -\sigma _{\nu }^{GDE1} & -\sigma _{\nu }^{GDE2} & 0 & 0 & 0 & 0 & 0 & 0
\\ 
0 & 0 & -\sigma _{\nu }^{GDE2} & 0 & 0 & 0 & 0 & 0 & 0 \\ 
0 & 0 & 0 & -\sigma _{\nu }^{GDI1} & -\sigma _{\nu }^{GDI2} & 0 & 0 & 0 & 0
\\ 
0 & 0 & 0 & 0 & -\sigma _{\nu }^{GDI2} & 0 & 0 & 0 & 0 \\ 
0 & 0 & 0 & 0 & 0 & \sigma _{\zeta }^{GDE1} & 0 & 0 & 0 \\ 
0 & 0 & 0 & 0 & 0 & 0 & \sigma _{\zeta }^{GDE2} & 0 & 0 \\ 
0 & 0 & 0 & 0 & 0 & 0 & 0 & \sigma _{\zeta }^{GDI1} & 0 \\ 
0 & 0 & 0 & 0 & 0 & 0 & 0 & 0 & \sigma _{\zeta }^{GDI2}%
\end{bmatrix}
\\
\veta_{t} &=&\left[ e_{t},\nu _{t}^{GDE1},\nu _{t}^{GDE2},\nu
_{t}^{GDI1},\nu _{t}^{GDI2},\zeta _{t}^{GDE1},\zeta _{t}^{GDE2},\zeta
_{t}^{GDI1},\zeta _{t}^{GDI2}\right] ^{\prime }
\end{eqnarray*}

\subsection{Identification and GDP$^{+}$}

Aruoba et al. (2016)\nocite{Aruobaetal2016} consider the problem of
identification in a special case of the GDE/GDI example considered above
where only a single vintage is available $\left( l=1\right) $. Their
unrestricted model may be written as\footnote{%
See Aruoba et al. (2016)\nocite{Aruobaetal2016}, equations (A.1) and (A.2).
Their model further differs from the model above in that (a) they model only the
sum of news and noise shocks, and (b) they assume that $\mT_{3}=\mT_{4}=0$, a condition that we will also impose, below.} 
\begin{eqnarray}
\begin{bmatrix}
GDE_{t} \\ 
GDI_{t}%
\end{bmatrix}
&=&%
\begin{bmatrix}
1 & 0 & 1 & 0 \\ 
1 & 0 & 0 & 1%
\end{bmatrix}%
\cdot 
\begin{bmatrix}
\tilde{y}_{t} \\ 
\tilde{y}_{t-1} \\ 
\eta _{t}^{E} \\ 
\eta _{t}^{I}%
\end{bmatrix}
\\
\begin{bmatrix}
\tilde{y}_{t+1} \\ 
\tilde{y}_{t} \\ 
\eta _{t+1}^{E} \\ 
\eta _{t+1}^{I}%
\end{bmatrix}
&=&%
\begin{bmatrix}
\rho  & 0 & 0 & 0 \\ 
1 & 0 & 0 & 0 \\ 
0 & 0 & 0 & 0 \\ 
0 & 0 & 0 & 0%
\end{bmatrix}%
\cdot 
\begin{bmatrix}
\tilde{y}_{t} \\ 
\tilde{y}_{t-1} \\ 
\eta _{t}^{E} \\ 
\eta _{t}^{I}%
\end{bmatrix}
 +%
\begin{bmatrix}
\sigma _{yy} & \sigma _{yE} & \sigma _{yI} \\ 
0 & 0 & 0 \\ 
\sigma _{Ey} & \sigma _{EE} & \sigma _{EI} \\ 
\sigma _{Iy} & \sigma _{IE} & \sigma _{II}%
\end{bmatrix}%
\cdot 
\begin{bmatrix}
e_{t}^{y} \\ 
e_{t}^{E} \\ 
e_{t}^{I}%
\end{bmatrix}%
\end{eqnarray}%
and they show that it is not identified. They propose adding a third
(instrumental) variable which is correlated with $\widetilde{y}_{t}$ but not
with $\eta _{t}^{E}$ or $\eta _{t}^{I}$, suggesting that household survey
data may be suitable for this purpose. We argue that the model may be
identified instead by increasing the number of vintages analysed and
assuming that measurement errors are the sum of news and noise measurement
errors as characterized above. We explore this point in the remainder of this
section by comparing the available number of sample moments to the number of
free parameters in the model. In the Appendix we provide a more rigorous
proof of identification in a slightly simpler model using the methods of
Komunjer and Ng (2011).

The essential insight comes from the form of the $\mR$ matrix in (\ref%
{eq:Rpartitioned}). News and noise measurement errors have tightly
constrained behaviour across successive data vintages; Noise errors are
assumed to be uncorrelated across vintages and with innovations in true
values, while news errors must be correlated with one another, with
innovations in true values, and their variances must be decreasing as series
are revised. 

If we have two series to reconcile (here $GDE$ and $GDI$) and $l$ vintages of
each, we have $2\cdot l\cdot (2\cdot l+1)/2$ observable cross moments as
well as $2\cdot l$ first-order autocorrelation coefficients, for a total of $%
l\cdot (2\cdot l+3)$ moments. The only free parameters in the above model,
however, are the autocorrelation coefficient $\rho $ and the $(1+4\cdot l)$
non-zero elements of $R$, for a total of $2\cdot (1+2\cdot l).$ This implies
that the number of available moments increases with $l^{2}$ while the number
of free parameters increases only with $l$.\footnote{%
Note that we have ignored any free parameters in $\mT_{3}$ and $\mT_{4}$ in
these calculations. We return to this, below. One must also keep in mind
that identification by data revision requires that the data are in fact
revised. If not, we effectively return to the underidentified case of $l=1$. 
}

In the special case where we use only a single data release, $l=1$, we
have $2\cdot (1+2\cdot 1)=6$ free parameters to estimate, but only $1\cdot
(2\cdot 1+3)=5$ available moments with which to do so. This is consistent
with the lack of identification noted by Aruoba et al. (2016)\nocite%
{Aruobaetal2016}. However, if we use $l=2$ data vintages, we have $2\cdot
(1+2\cdot 2)=10$ free parameters and $2\cdot (2\cdot 2+3)=14$ moments with which to
identify them. For $l=3$ we have 27 moments with which to estimate 14
parameters and for $l=4$ (the case we consider below) we have 44 moments
with which to estimate 18 parameters.

This suggests that as we add more data releases, we potentially have the
ability to generalize the model further still. The univariate data revision
model of JvN envisages two such types of generalization.

\begin{enumerate}
\item We may wish to relax some of the zero restrictions on $\mR$. In particular, it may be desirable to allow for news shocks to
be correlated across the two variables, or to allow for noise shocks to be
correlated across data releases. 

\item We may wish to relax some of the zero restrictions on the transition
matrix in (\ref{eq:Transition_new}) to allow for measurement errors to be
correlated across calendar periods. (JvN refer to these as \textquotedblleft
spillover\textquotedblright\ effects.)
\end{enumerate}

In the Appendix, we briefly explore the possibilities for identification
with some of these generalizations. We now turn to consider the revisions in
the available data.

\section{Data and Estimation\label{data_estimation}}
\subsection{Data}
We use monthly vintages of quarterly expenditure-based and income-based estimates of GDP from the Bureau of Economic Analysis (BEA) covering the period 2003Q1--2014Q3. 
For $GDE$ we employ the Advance, the Third, the 12th and the 24th releases and Second/Third, 12th and the 24th releases for $GDI$. Due to a lag in source data availability the BEA does not prepare Advance estimates for $GDI$. The initial estimates for $GDI$ are presented with the Second $GDI$ estimate. Estimates for fourth quarter $GDI$ are presented in the Third estimate only.\footnote{See Fixler et al. (2014)\nocite{FixlerGreenawayGrimm2014} for a more detailed discussion of the $GDE$-$GDI$ vintage history.} 


\subsection{Estimation}
We employ Gibbs Sampling methods to obtain posterior simulations for our model's parameters (see, e.g.,  Kim and Nelson 1999\nocite{KimNelson1999}). We use conjugate and diffuse priors for the coefficients and the variance covariance matrix, resulting in a multivariate normal posterior for the coefficients and an inverted Wishart posterior for the variance covariance matrix. For the prior for the coefficients restricted to zero we assume the mean to be zero and variance to be close to zero. 

Our Gibbs sampler has the following structure. We first initialize the sampler with values for the coefficients and the variance covariance matrix. Conditional on data, the most recent draw for the coefficients and for the variance covariance matrix, we draw the latent state variables $\valpha_t$ for $t=1,...,T$ using the procedure described in Carter and Kohn (1994)\nocite{CarterKohn1994}. In the next step, we condition on data, the most recent draw for the latent variable $\valpha_t$ and for the variance covariance matrix, drawing the coefficients from a multivariate normal distribution. Finally, conditional on data, the most recent draw for the latent variables and the coefficients, we draw the variance covariance matrix from an inverted Wishart distribution. We cycle through 100K Gibbs iterations, discarding the first 90K as burn-in. Of those 10K draws we save only every 10th draw, which gives us in total 1000 draws on which we base our inference. Convergence of the sampler was checked by studying recursive mean plots and by varying the starting values of the sampler and comparing results.

\section{Results\label{results}}
Here we compare our measure of $GDP$ to releases of $GDE$ and $GDI$ in four different ways: (i) in graphs, (ii) looking at historical decompositions, (iii) by investigating dynamics, and (iv) by calculating relative contributions. To distinguish between the true unknown values of $GDP$ and our model's estimates of these values, we refer to our model's estimates as $GDP^{++}$.

\begin{figure}[hp]
    \caption{$GDP^{++}$ vs. $GDE$\label{GDP++_GDPE}}
    \begin{center}
    \includegraphics[width=\textwidth]{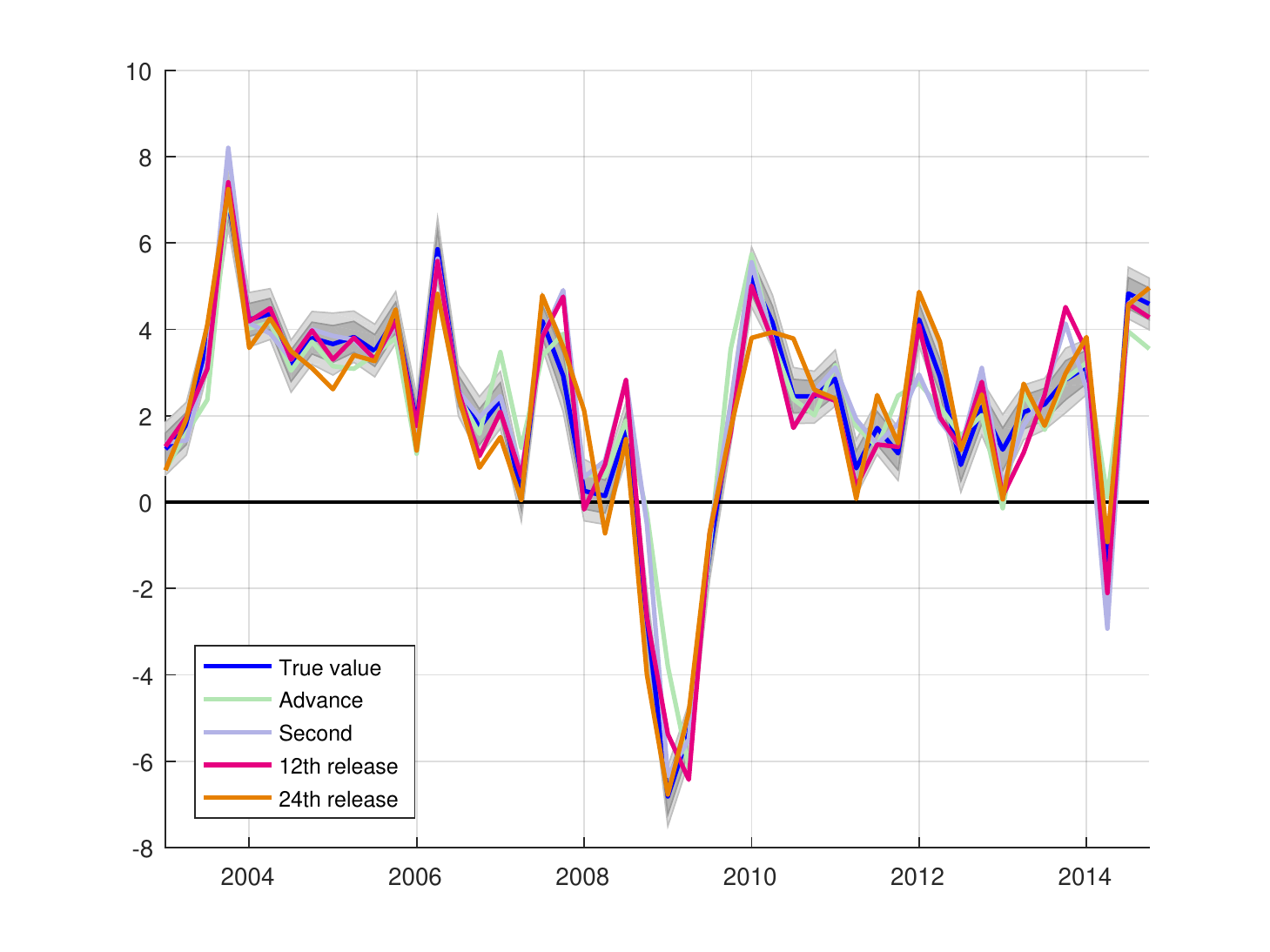}
    \end{center}
{\footnotesize The blue line represents the posterior mean of $GDP$ (the ``true'' value) and the shaded area around the blue line indicates 90\% of posterior probability mass. The green line represents the advance estimate, the purple line is the second estimate, the red line the 12th release and the orange line the 24th release of expenditure side $GDP$ growth.}
\end{figure}
 
\subsection{Comparison of $GDP^{++}$ and releases of $GDE$ and $GDI$}
In Figure~\ref{GDP++_GDPE} we compare $GDP^{++}$ and its shaded posterior ranges (90\% of probability mass) to the four releases of $GDE$ we employed in the estimation, the Advance, third, the 12th and the 24th release. There is some evidence that the releases are more volatile than the true values of $GDP$. We observe that the releases are outside the posterior bounds for some periods. This observation holds especially for the Advance release and the 24th release; in some periods, like e.g. 2010Q1, the Advance release and the 24th release are on different sides of the posterior range.  

\begin{figure}[ht]
\caption{$GDP^{++}$ vs. $GDI$\label{GDP++_GDPI}}
\begin{center}
    \includegraphics[width=\textwidth]{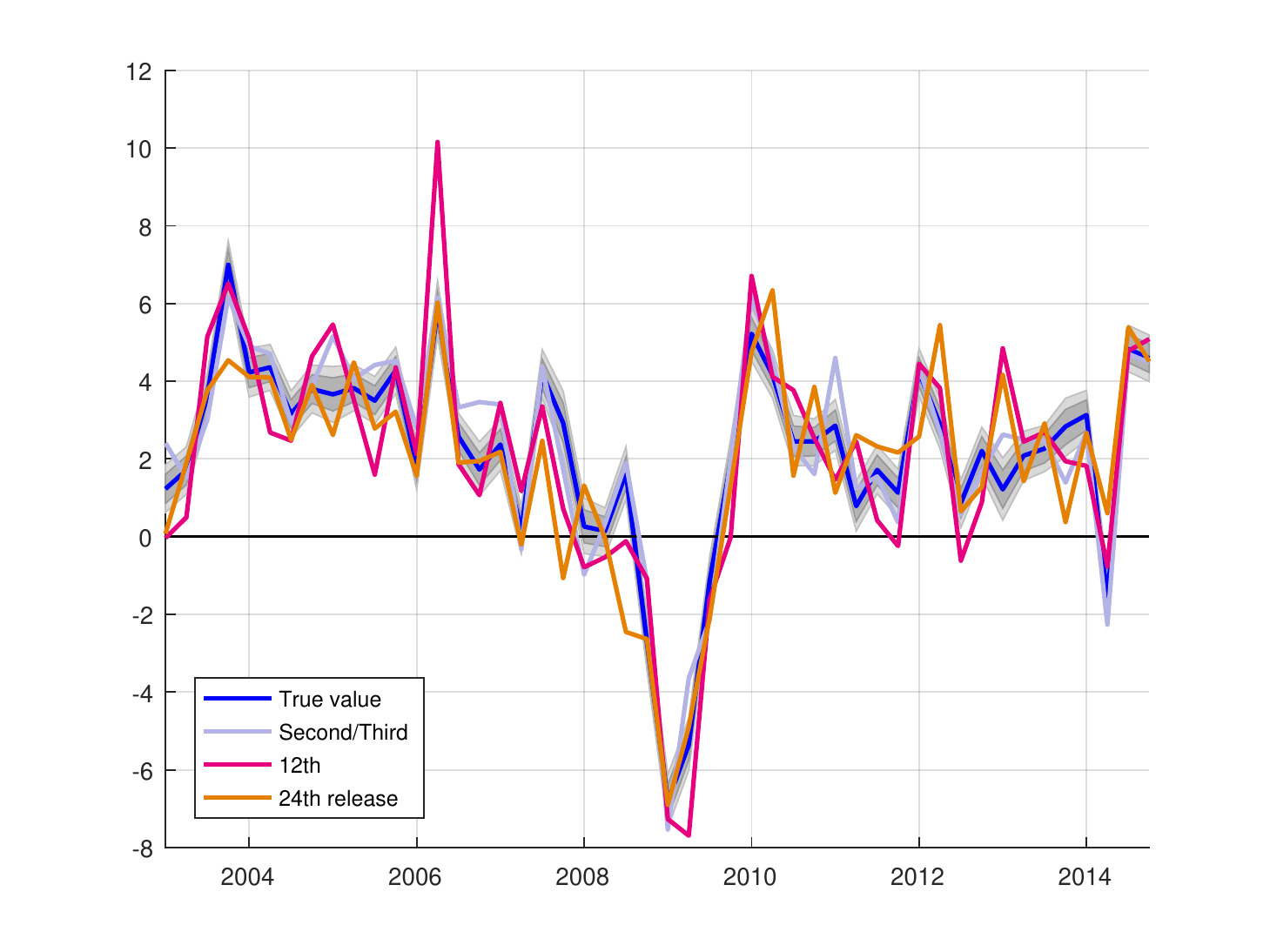}
\end{center}    
{\footnotesize The blue line represents the posterior mean of $GDP$, the ``true'' value, and the shaded area around the blue line indicates 90\% of posterior probability mass. The purple line is the second/third estimate, the red line the 12th release and the orange line the 24th release of income side $GDP$ growth.}
\end{figure}

Figure~\ref{GDP++_GDPI} shows $GDP^{++}$ together with shaded posterior ranges (90\% of probability mass) and the three releases of $GDI$ we employed in the estimation, the Second/Third, the 12th and the 24th release. The releases fluctuate around the posterior bounds of the true values. The $GDI$ releases are more volatile than our estimates $GDP^{++}$. The releases of $GDI$ are also much more volatile than the releases of $GDE$. 
Note that the sample paths of $GDP_M$ and $GDE$ and $GDI$ in Aruoba et al. (2016, Figure 3) show a different picture than our Figures \ref{GDP++_GDPE} and \ref{GDP++_GDPI}. $GDE$ differs more from their $GDP$ measure than $GDI$.  

\subsection{Historical decomposition\label{historical_decomposition}}
Our econometric framework (\ref{eq:Measurement_new}-\ref{eq:Transition_new}) allows the historical decomposition of $GDE$ and $GDI$ in terms of news and noise measurement errors. We illustrate the decomposition for $GDE$. 

Suppose, we have $l$ releases of $GDE_t$
\begin{align}
GDE_t^1 &= \rho GDP_{t-1}^{++} + \eta_{Gt}+ \eta_{E \zeta t}^1 \nonumber\\
GDE_t^2 &= \rho GDP_{t-1}^{++} + \eta_{Gt}+  \eta_{E\nu t}^1 + \eta_{E \zeta t}^2 \nonumber\\
\vdots \ \     &= \ \ \ \ \ \ \ \ \ \  \vdots \nonumber\\
GDE_t^l &= \rho GDP_{t-1}^{++} +  \eta_{Gt} + \eta_{E \nu t}^1 + ... + \eta_{E\nu t}^{l-1} + \eta_{E \zeta t}^l. \nonumber
\end{align}
Then the total revision of $GDE$ can be written as 
\begin{align}
GDE_t^l - GDE_t^1&= \underbrace{\eta_{E \nu t}^1 + ... + \eta_{E\nu t}^{l-1}}_\text{News} + \underbrace{\eta_{E \zeta t}^l - \eta_{E \zeta t}^1}_\text{Noise} \label{histdecomp}
\end{align}
where every element on the right-hand side of the equation is part of the state vector whose estimates may be recovered using standard techniques.

\begin{figure}[ht]
\caption{Historical Decomposition\label{fig:historical_decompositions}}
\begin{center}
$GDE$ \\
\includegraphics[width=.75\textwidth]{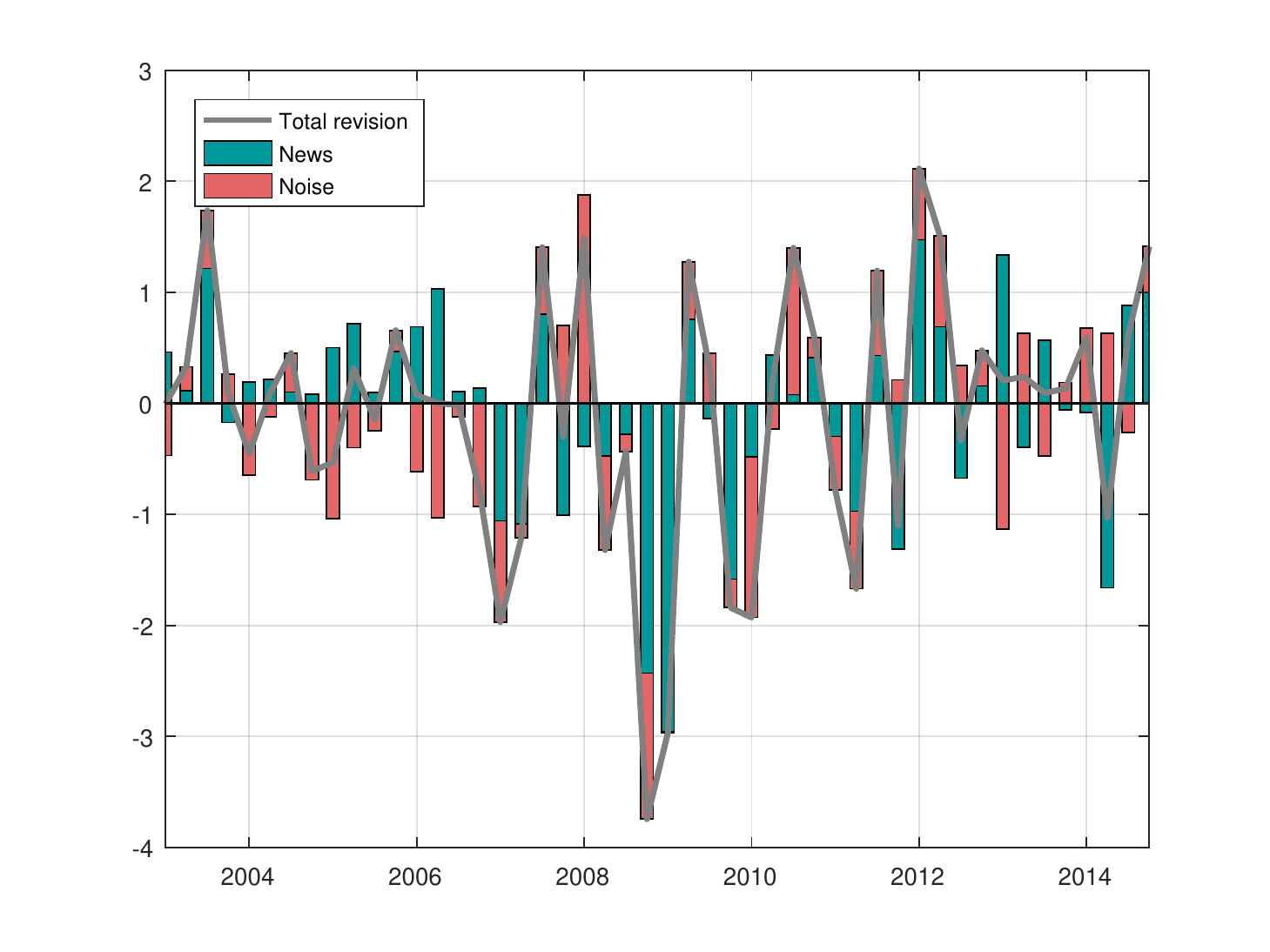} \\
$GDI$ \\
\includegraphics[width=.75\textwidth]{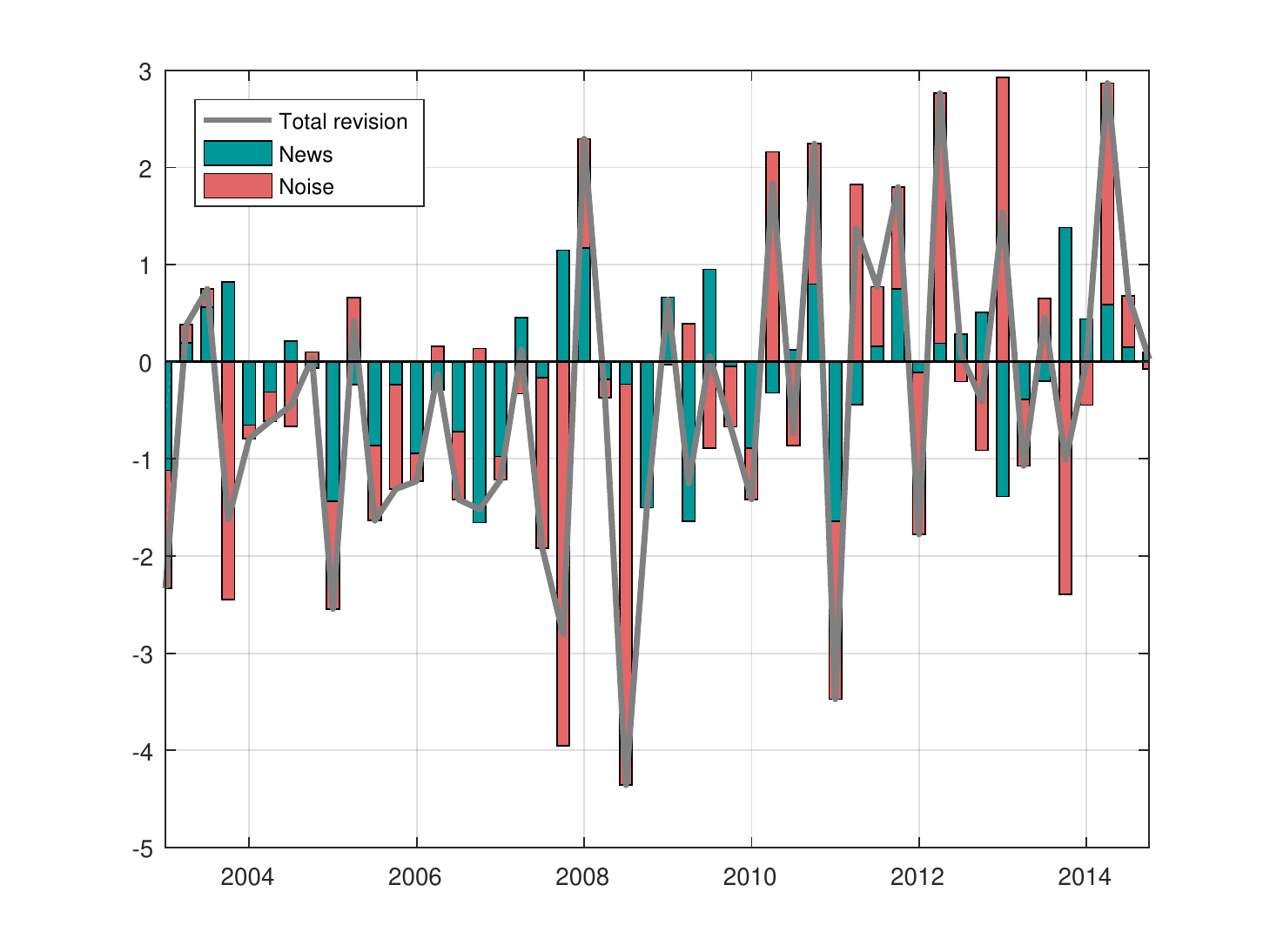} 
\end{center}
{\footnotesize Historical decomposition of the total revision (24th release minus second estimate) into news and noise. The red bars depict the share of news and the green bars the share of noise in total revision (grey line). The historical decomposition is based on the decomposition described in (\ref{histdecomp}).}
\end{figure}
\clearpage

The outcomes of the historical decompositions are shown in Figure \ref{fig:historical_decompositions}. 
The top panel shows total revisions in $GDE$ with news and noise shares, the bottom panel total $GDE$ revisions with news and noise shares. 
We observe that total revisions in $GDI$, the bottom panel, are larger than total revisions in $GDE$, a stylized fact which can also be distilled from the previous two figures. The two panels suggest that the news share in total $GDE$ revisions is larger than the noise share while the opposite seems to hold for total revisions in $GDI$. This observation is consistent with Fixler and Nailewaik (2009), who also reject the pure noise assumption in $GDI$. It also appears that GDI was particularly noisy around the start of 2008 and after 2012.

\subsection{Dynamics of $GDP^{++}$ and other GDP measures}

In Figure \ref{fig:gdp_dynamics} we depict the ($\rho,\sigma^2$) pairs summarizing the dynamics of our true $GDP$ estimate across all draws. We contrast the ($\rho,\sigma^2$) pairs corresponding to our $GDP^{++}$ estimate to the ($\rho,\sigma^2$) pairs obtained when using a news measurement error only or a noise measurement error only version of our model, the benchmark model estimated in Aruoba et al. (2016) and when fitting an AR(1) model to $GDE$ and $GDI$.

\begin{figure}[ht]
\caption{GDP Dynamics}\label{fig:gdp_dynamics}
\begin{center}
\includegraphics[width=.8\linewidth]{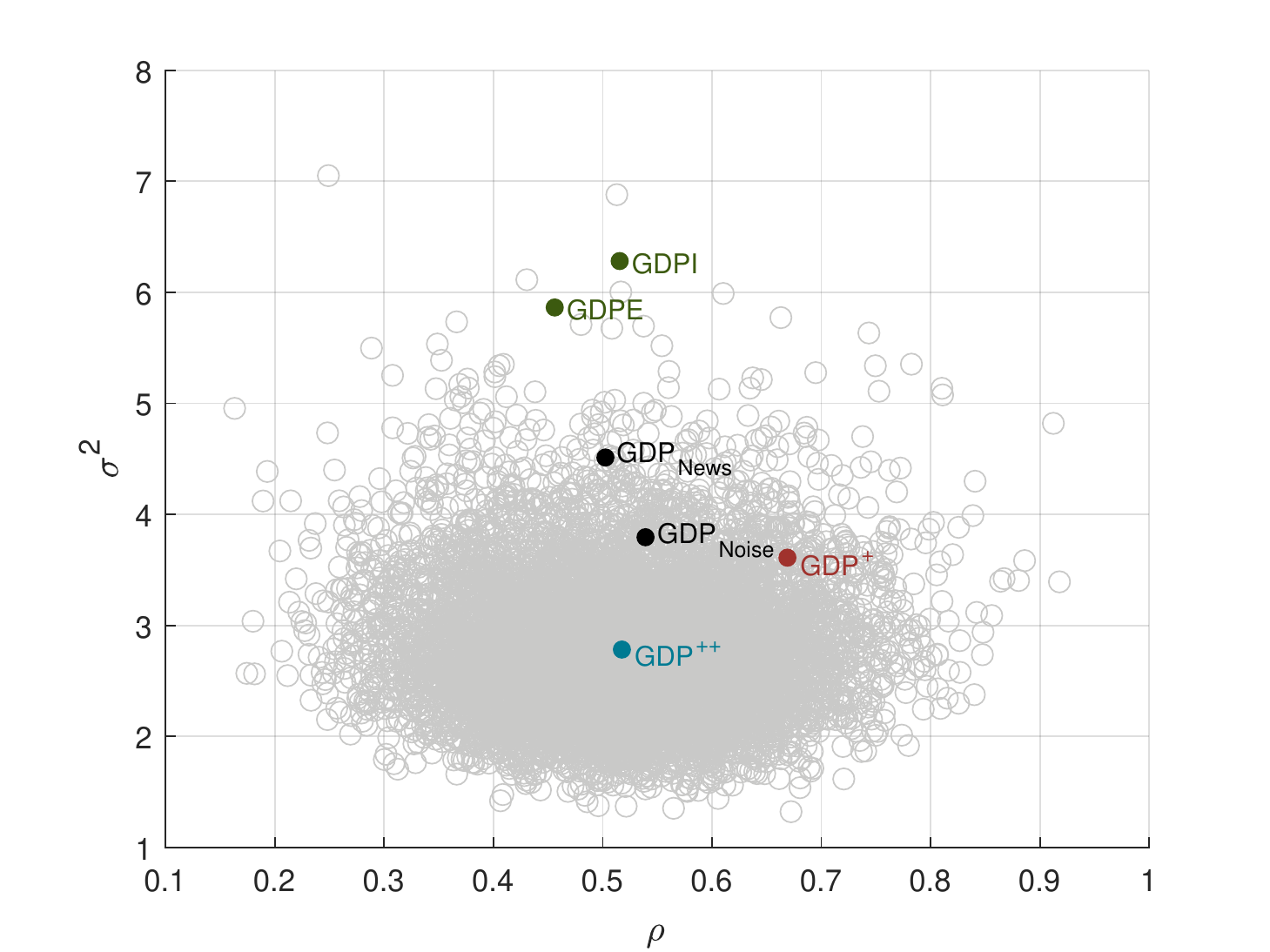}
\end{center}
{\footnotesize The grey shaded area consists of ($\rho, \sigma^2$) pairs across draws from our sampler and the blue  dot is the posterior mean of the ($\rho,\sigma^2$) pairs across draws. The black dots represent the posterior mean of the ($\rho,\sigma^2$) pairs of the news only and noise model, respectively. The red dot is the posterior mean of the ($\rho,\sigma^2$) pairs of $GDP^+$ using the benchmark specification ($\zeta=0.8$) described in Aruoba et al. (2016). The green dots are ($\rho,\sigma^2$) pairs, resulting from AR(1) models fitted to expenditure side and income side GDP growth, respectively. The sampling period for re-estimating the Aruoba et al. (2016) model and for fitting the AR(1) models to the two GDP measures is 2003Q3--2014Q3 (released on October 28, 2016).}
\end{figure}
~

Figure \ref{fig:gdp_dynamics} reveals that our real-time data based estimate of $GDP$ is somewhat less persistent than the $GDP^+$ measure of Aruoba et al. (2016), but exhibits a higher persistence than the estimates for $GDE$ and $GDI$.\footnote{We thank Dongho Song for making his Matlab code available online.} We also find that the posterior mean of the innovation variance of our $GDP^{++}$ is much smaller than the innovation variances of $GDE$, $GDI$ and the benchmark model of Aruoba et al. (2016). The innovation variance of $GDP^{++}$ is also smaller than the innovation variance of the models estimated with news and noise measurement errors only, which in turn are higher than the innovation variance of $GDP^+$. The combination of a $\rho$ that is close to those implied by the various models estimated in Aruoba et al. (2016) and a $\sigma^2$ that is much smaller than the ones implied by Aruoba et al. (2016) leads to a higher forecastability of the $GDP^{++}$ measure.  
 
\subsection{Relative contributions of $GDE$ and $GDI$ to $GDP^{++}$}

To assess the relative importance of $GDI$ and $GDE$ at different releases,  we use the Kalman gains. 
They represent the weight that the estimated value places on estimates of various releases.
The outcomes are listed in Table \ref{tab:KalmanGain}. 

\begin{table}[htb]
  \centering
  \caption{Kalman Gains
  \label{tab:KalmanGain}}
    \begin{tabular}{lrr}
    ~\\
    \toprule
    Weight on &  $GDE$ & $GDI$   \\
    \midrule
    \textbf{News and Noise} &       &  \\
    Advance    & 0.0474 &  \\
    Second/Third    & 0.0402 & 0.2624 \\
    12th   & 0.2879 & 0.0473 \\
    24th Release   & 0.2824 & 0.079 \\
    \midrule
     \textbf{News Only }&       &  \\
    Advance          & 0.0073 &  \\
    Second/Third     & 0.0000 & 0.0073 \\
    12th             & 0.0000 & 0.0000 \\
    24th Release     & 0.9786 & 0.0142 \\
    \midrule
    \textbf{Noise Only} &       &  \\
    Advance          & 0.1421 &  \\
    Second/Third     & 0.2823 & 0.0688 \\
    12th             & 0.3309 & 0.0303 \\
    24th Release     & 0.0997 & 0.0283 \\
    \bottomrule
    \end{tabular}%
  \label{tab:KalmanGains}%
\end{table}%


 
 
The results show that the weights assigned to different releases vary greatly as we change the assumed structure of the measurement error. When they are assumed to be pure News, the second panel of the table shows that 98\% of the weight is put on the last release of $GDE$. Once we allow for the possibility of noise errors, however, more weight is assigned to $GDI$ and weights are spread over more releases. The earliest releases of $GDE$ receive less weight than the later releases, while the opposite is true for $GDI$.  In all cases, we also find that $GDE$ releases are more important for explaining $GDP$ than $GDI$ releases, in contrast to Aruoba et al. (2016).

\section{Conclusion\label{conclusion}}

We have described a new approach to data reconciliation that exploits multiple data releases on each series. This helps both with the identification of measurement errors and with optimally extracting information from multiple noisy series. We used this to propose
a new measure of U.S. $GDP$ growth using real-time data on $GDE$ and $GDI$.  Our measure $GDP^{++}$ is shown to be more persistent than $GDE$ and $GDI$ and has  smaller residual variance. In addition it has a similar autoregressive coefficient but smaller residual variance than the $GDP$ measure $GDP^+$ of Aruoba et al. (2016). Historical decompositions of $GDE$ and $GDI$ measurement errors reveal a larger news share in $GDE$ than in $GDI$. 

\clearpage
\appendix 
\setcounter{section}{1}
\renewcommand{\thesection}{\Alph{section}}
\setcounter{equation}{0}
\renewcommand{\theequation}{\Alph{section}.\arabic{equation}}

\section*{Appendix} 
This appendix first analyzes the identification of the univariate state space system in Jacobs and van Norden (2011)\nocite{JacobsvanNorden2011}, using the procedure described in Komunjer and Ng (2011)\nocite{KomunjerNg2011} and used by Aruoba et al. (2016)\nocite{Aruobaetal2016}. Thereafter we discuss the possibilities for identification in more general reconciliation models by comparing the number of free parameters to the number of available moment conditions. 

\section*{Identification in the univariate, two vintage JvN framework}
The state space form of the Jacobs and van Norden (JvN) model with two vintages and no spillovers can be expressed as
\begin{align}
\begin{bmatrix}
y_{t}^{1} \\
y_{t}^{2}%
\end{bmatrix}&=%
\begin{bmatrix}
1 & 1 & 0 & 1 & 0 \\ 1 & 0 & 1 & 0  & 1
\end{bmatrix}%
\begin{bmatrix}
\tilde{y}_t \\ \nu_t^{1} \\ \nu_t^{2}\\ \zeta_t^{1} \\ \zeta_t^{2}
\end{bmatrix} \label{eq:measure0},\\
\begin{bmatrix}
\tilde{y}_t \\ \nu_t^{1} \\ \nu_t^{2}\\ \zeta_t^{1} \\ \zeta_t^{2}
\end{bmatrix}&=%
\begin{bmatrix}
\rho & 0 & 0 & 0 & 0 \\ 0 & 0 & 0 & 0 & 0 \\ 0 & 0 & 0 & 0 & 0 \\ 0 & 0 & 0 & 0 & 0 \\ 0 & 0 & 0 & 0 & 0
\end{bmatrix}%
\begin{bmatrix}
\tilde{y}_{t-1} \\ \nu_{t-1}^{1} \\ \nu_{t-1}^{2}\\ \zeta_{t-1}^{1} \\ \zeta_{t-1}^{2}
\end{bmatrix} +%
\begin{bmatrix}
1 & 1 & 1 & 0 & 0 \\ 0 & -1 & -1 & 0 & 0 \\ 0 & 0 & -1 & 0 & 0 \\ 0 & 0 & 0 & 1 & 0 \\ 0 & 0 & 0 & 0 & 1
\end{bmatrix}%
\begin{bmatrix}
\eta_{t,\tilde{y}} \\ \eta_{t,\nu}^{1} \\ \eta_{t,\nu}^{2} \\ \eta_{t,\zeta}^{1} \\ \eta_{t,\zeta}^{2}
\end{bmatrix}, \label{eq:state0}
\end{align}
where $y_t^{i}$ for $i=1,2$ denotes the different releases, $\tilde{y}_t$ is the ``true'' value of the variable of interest, $\nu_t^{i}$ and $\zeta_t^{i}$ for $i=1,2$ are the news and the noise components and $\eta_{t,\nu}^{i}$ and $\eta_{t,\zeta}^{i}$ for $i=1,2$ are the news and the noise shocks, $[\eta_{t,\tilde{y}} \ \eta_{t,\nu}^{1} \ \eta_{t,\nu}^{2} \ \eta_{t,\zeta}^{1} \ \eta_{t,\zeta}^{2}]' \sim N(\bm{0},\mH)$ with $\mH = \diag(\sigma_{\tilde{y}}^2, \sigma_{\nu1}^2, \sigma_{\nu2}^2, \sigma_{\zeta1}^2, \sigma_{\zeta2}^2)$, where $\diag$ denotes a diagonal matrix.

\newpage
The system in (\ref{eq:measure0}) and (\ref{eq:state0}) can also be written as

\begin{align}
\begin{bmatrix}
y_{t}^{1} \\
y_{t}^{2}%
\end{bmatrix}&=%
\begin{bmatrix}
1 \\ 1
\end{bmatrix}%
\tilde{y}_t +%
\begin{bmatrix}
0 & 1 & 0 & 1 & 0 \\ 0 & 0 & 1 & 0  & 1
\end{bmatrix}%
\begin{bmatrix}
\omega_{t,\tilde{y}} \\ \omega_{t,\nu}^{1} \\ \omega_{t,\nu}^{2} \\ \omega_{t,\zeta}^{1} \\ \omega_{t,\zeta}^{2}
\end{bmatrix},\label{eq:measure1}\\
\tilde{y}_t&=%
\rho \tilde{y}_{t-1} +
\begin{bmatrix}
1 & 0 & 0 & 0 & 0
\end{bmatrix}%
\begin{bmatrix}
\omega_{t,\tilde{y}} \\ \omega_{t,\nu}^{1} \\ \omega_{t,\nu}^{2} \\ \omega_{t,\zeta}^{1} \\ \omega_{t,\zeta}^{2}
\end{bmatrix},\label{eq:state1}
\end{align}%
where $\omega_{t,\tilde{y}} = \eta_{t,\tilde{y}} + \eta_{t,\nu}^{1} + \eta_{t,\nu}^{2}$, $\omega_{t,\nu}^{1} = - \eta_{t,\nu}^{1} - \eta_{t,\nu}^{2}$, $\omega_{t,\nu}^{2} = -\eta_{t,\nu}^{2}$, $\omega_{t,\zeta}^{1} =\eta_{t,\zeta}^{1}$, $\omega_{t,\zeta}^{2} =\eta_{t,\zeta}^{2}$ and
$[\omega_{t,\tilde{y}} \ \omega_{t,\nu}^{1} \ \omega_{t,\nu}^{2} \ \omega_{t,\zeta}^{1} \ \omega_{t,\zeta}^{2}]' \sim N(\bm{0},\mSigma)$ with variance-covariance matrix $\mSigma$ defined as
\begin{equation}\label{def:Sigma}
\mSigma =  \begin{bmatrix}
\Sigma_{\tilde{y}\tilde{y}} & \Sigma_{\tilde{y}\nu1} & \Sigma_{\tilde{y}\nu2} & 0 & 0 \\
\Sigma_{\nu1\tilde{y}} & \Sigma_{\nu1\nu1} & \Sigma_{\nu1\nu2} & 0 & 0 \\
\Sigma_{\nu2\tilde{y}} & \Sigma_{\nu2\nu1} & \Sigma_{\nu2\nu2} & 0 & 0 \\
0 & 0 & 0 & \Sigma_{\zeta1\zeta1} & 0 \\
0 & 0 & 0 & 0 & \Sigma_{\zeta2\zeta2} \\
\end{bmatrix},
\end{equation}

where
\begin{align}
\Sigma_{\tilde{y}\tilde{y}}&= \sigma_{\tilde{y}}^2+\sigma_{\nu1}^2+\sigma_{\nu2}^2, \ \
\Sigma_{\tilde{y}\nu1}= -\sigma_{\nu1}^2-\sigma_{\nu2}^2,\nonumber\\
\Sigma_{\tilde{y}\nu2}&= -\sigma_{\nu2}^2, \ \
\Sigma_{\nu1\nu1}= \sigma_{\nu1}^2+\sigma_{\nu2}^2,\ \
\Sigma_{\nu1\nu2} = \sigma_{\nu2}^2,\label{eq:Sigdef}\\
\Sigma_{\nu2\nu2}&= \sigma_{\nu2}^2,\ \
\Sigma_{\zeta1\zeta1}= \sigma_{\zeta1}^2,\ \
\Sigma_{\zeta2\zeta2}= \sigma_{\zeta2}^2,\nonumber
\end{align}
which implies
\begin{align}
\Sigma_{\tilde{y}\nu1}= -\Sigma_{\nu1\nu1},\nonumber\\
\Sigma_{\tilde{y}\nu2}= -\Sigma_{\nu2\nu2},\label{eq:Sigequal}\\
\Sigma_{\nu2\nu2}= \Sigma_{\nu1\nu2}.\nonumber
\end{align}
Moreover, consider the following restriction
\begin{equation}\label{eq:rest1}
\sigma_{\nu2}^2 = 0,
\end{equation}
which is justified due to the fact that the second release news shock corresponds to information outside the sample and is thus not needed.

Aruoba et al. (2016)\nocite{Aruobaetal2016} show that a state space system described in Equations (\ref{eq:measure1}) and (\ref{eq:state1}) is not identified with $\mSigma$ unrestricted and identified under certain restrictions on elements of $\mSigma$. We now investigate whether the restrictions implied by JvN lead to an identified system following the procedure described in Aruoba et al. (2016)\nocite{Aruobaetal2016}.

\begin{theorem}\label{thm: main}
Suppose that Assumptions 1, 2, 4-NS and 5-NS of Komunjer and Ng (2011)\nocite{KomunjerNg2011} hold. Then according to Proposition 1-NS of Komunjer and Ng (2011)\nocite{KomunjerNg2011}, the state space model described in (\ref{eq:measure0}) and (\ref{eq:state0}) is identified given the restrictions implied by (\ref{eq:measure0}), (\ref{eq:state0}) and (\ref{eq:rest1}).
\end{theorem}

\begin{proof}[Proof of Theorem \ref{thm: main}]
We begin by rewriting the state space model in (\ref{eq:measure1}) and (\ref{eq:state1}) to match the notation used in Komunjer and Ng (2011)\nocite{KomunjerNg2011}
\begin{align}
x_{t+1} & = A(\theta) x_t + \bm{B(\theta)} \bm{\epsilon}_{t+1} \label{eq:stateKN}\\
\vz_{t+1} &= \bm{C(\theta)} x_t + \bm{D(\theta)} \bm{\epsilon}_{t+1}\label{eq:measureKN},
\end{align}
where $x_{t}= \tilde{y}_t$, $\vz_t = [y_{t}^{1} \ y_{t}^{2}]'$ , $\bm{\epsilon}_{t} = [\omega_{t,\tilde{y}} \ \omega_{t,\nu}^{1} \ \omega_{t,\nu}^{2} \ \omega_{t,\zeta}^{1} \ \omega_{t,\zeta}^{2}]' $, $A(\theta)=\rho$, $\bm{B(\theta)} = [1 \ 0 \ 0 \ 0 \ 0]$, $\bm{C(\theta)}= [\rho \ \rho]'$,
$$ \bm{D(\theta)} = \begin{bmatrix}
1 & 1 & 0 & 1 & 0 \\
1 & 0 & 1 & 0 & 1
\end{bmatrix}$$
and $\vtheta = [\rho \ \sigma_{\tilde{y}}^2 \ \sigma_{\nu1}^2 \ \sigma_{\nu2}^2 \ \sigma_{\zeta1}^2 \ \sigma_{\zeta2}^2]'$.

Given that $\mSigma$ is positive definite and $0 \leq \rho<1$, Assumption 1 and 2 of Komunjer and Ng (2011)\nocite{KomunjerNg2011} are satisfied. Given that $\bm{D(\theta)} \mSigma \bm{D(\theta)}'$ is nonsingular also Assumption 4-NS of Komunjer and Ng (2011)\nocite{KomunjerNg2011} is satisfied. Rewriting the state space model in (\ref{eq:stateKN}) and (\ref{eq:measureKN}) into its innovation representation gives
\begin{align}
\hat{x}_{t+1|t+1} & = A(\theta) \hat{x}_{t|t} + \bm{K(\theta)} \va_{t+1}\\
\vz_{t+1} &= \bm{C(\theta)} \hat{x}_{t|t}  + \va_{t+1},
\end{align}
where $\bm{K(\theta)}$ is the Kalman gain and $\va_{t+1}$ is the one-step ahead forecast error of $z_{t+1}$ with variance $\bm{\Sigma_a(\theta)}$. The Kalman gain and the variance of the one-step ahead forecast error for this system can be expressed as
\begin{align}
\bm{K(\theta)}& = (p \rho \mC' + \bm{\Sigma_{BD}})(p \mC \mC' + \bm{\Sigma_{DD}})^{-1} \label{eq:Ktheta}\\
\bm{\Sigma_a(\theta)}& = p \mC \mC' + \bm{\Sigma_{DD}},  \label{eq:Sigtheta}
\end{align}
where $p$ is the variance of the state vector, solving the following Riccati equation
\begin{equation}\label{eq: riccati}
p = p \rho^2 + \Sigma_{BB} - (p \rho \mC' + \bm{\Sigma_{BD}})(p \mC \mC' + \bm{\Sigma_{DD}})^{-1}(p \rho \mC + \bm{\Sigma_{DB}}).
\end{equation}
and $\Sigma_{BB}= \mB \mSigma \mB'$, $\bm{\Sigma_{BD}}= \mB \mSigma \mD'$, $\bm{\Sigma_{DD}}= \mD \mSigma \mD'$ with

\begin{align}
\Sigma_{BB} & = \Sigma_{\tilde{y}\tilde{y}},\nonumber\\
\bm{\Sigma_{BD}} & = \begin{bmatrix}
\Sigma_{\tilde{y}\tilde{y}} + \Sigma_{\tilde{y}\nu1} & \Sigma_{\tilde{y}\tilde{y}} + \Sigma_{\tilde{y}\nu2}
\end{bmatrix},\label{eq:SigBD} \\
\bm{\Sigma_{DD}} & = \begin{bmatrix}
\Sigma_{\tilde{y}\tilde{y}} + 2\Sigma_{\tilde{y}\nu1} + \Sigma_{\nu1\nu1} + \Sigma_{\zeta1\zeta1} & . \nonumber \\
\Sigma_{\tilde{y}\tilde{y}} + \Sigma_{\tilde{y}\nu1} + \Sigma_{\nu2\tilde{y}} + \Sigma_{\nu2\nu1} & \Sigma_{\tilde{y}\tilde{y}} + 2\Sigma_{\tilde{y}\nu2} + \Sigma_{\nu2\nu2} + \Sigma_{\zeta2\zeta2}
\end{bmatrix}.
\end{align}
By using the definitions in (\ref{eq:Sigdef}), the expressions in (\ref{eq:SigBD}) can also be written as
\begin{align}
\Sigma_{BB} & = \sigma_{\tilde{y}}^2 + \sigma_{\nu1}^2 + \sigma_{\nu2}^2 , \ \ \ \ \
\bm{\Sigma_{BD}}= \begin{bmatrix}
\sigma_{\tilde{y}}^2 & \sigma_{\tilde{y}}^2 + \sigma_{\nu1}^2
\end{bmatrix},\nonumber \\
\bm{\Sigma_{DD}} & = \begin{bmatrix}
\sigma_{\tilde{y}}^2 + \sigma_{\zeta1}^2 & . \\
\sigma_{\tilde{y}}^2  & \sigma_{\tilde{y}}^2 + \sigma_{\nu1}^2 + \sigma_{\zeta2}^2
\end{bmatrix}. \label{eq:Sigsig}
\end{align}

Assumption 5-NS of Komunjer and Ng (2011)\nocite{KomunjerNg2011} relates to the \emph{controllability} and \emph{observability} of state space systems. The state space system in (\ref{eq:measure1}) and (\ref{eq:state1}) is \emph{controllable} if matrix $[ \bm{K(\theta)} \  A(\theta) \bm{K(\theta)}]$ has full row rank and it is \emph{observable} if the matrix $[\bm{C(\theta)}' \ A(\theta)' \bm{C(\theta)}']$ has full column rank and is thus said to be minimal.

To show that Assumption 5-NS  is satisfied, first note that $\Sigma_{BB}-\bm{\Sigma_{BD}} \bm{\Sigma_{DD}}^{-1}\bm{\Sigma_{DB}}$ is the Schur complement of $\mOmega$, the variance covariance matrix of the joint distribution of $x_{t+1}$ and $z_{t+1}$, with respect to $\bm{\Sigma_{DD}}$ where
$$\mOmega = \begin{bmatrix}
\Sigma_{BB}      & \bm{\Sigma_{BD}} \\
\bm{\Sigma_{DB}} & \bm{\Sigma_{DD}}
\end{bmatrix}.
$$ Because $\mOmega$ is a positive definite matrix, its Schur complement is also positive definite thus leading to $\Sigma_{BB}-\bm{\Sigma_{BD}} \bm{\Sigma_{DD}}^{-1} \bm{\Sigma_{DB}}>0$. Now to show that this inequality leads to $p>0$, we use the following lemma 

\begin{lemma}\label{lem1}
Assume $\mA$ and $(\mA+\mB)$ are invertible and that $rank(\mB)=1$, then $$ (\mA+\mB)^{-1} = \mA^{-1} - \frac{1}{1+tr(\mB \mA^{-1})}\mA^{-1}\mB\mA^{-1}.$$
\end{lemma}

We can now use Lemma \ref{lem1} to rewrite Equation (\ref{eq: riccati}) as
\begin{align}\label{eq: riccati2}
p & = p \rho^2 + \Sigma_{BB} - (p \rho \mC' + \bm{\Sigma_{BD}} ) \bm{\Sigma_{DD}}^{-1} (p \rho \mC + \bm{\Sigma_{DB}}) \nonumber \\
  & + \frac{p}{g}(p \rho \mC' + \bm{\Sigma_{BD}})\bm{\Sigma_{DD}}^{-1} \mC \mC'\bm{\Sigma_{DD}}^{-1}(p \rho \mC + \bm{\Sigma_{DB}} ),
\end{align}
where $g=1+p tr(\mC \mC' \bm{\Sigma_{DD}}^{-1})$. After some manipulations we find the following quadratic equation
\begin{equation}
a p^2 +  b p + c = 0,
\end{equation}
with
\begin{align}
a & =-tr(\mC \mC' \bm{\Sigma_{DD}}^{-1}),\nonumber \\
b & =(\rho-\bm{\Sigma_{BD}} \bm{\Sigma_{DD}}^{-1} \bm{\Sigma_{DB}})^2 + tr(\mC \mC' \bm{\Sigma_{DD}}^{-1}) (\Sigma_{BB} - \bm{\Sigma_{BD}} \bm{\Sigma_{DD}}^{-1} \bm{\Sigma_{DB}})-1,\nonumber \\
c & = \Sigma_{BB} - \bm{\Sigma_{BD}} \bm{\Sigma_{DD}}^{-1} \bm{\Sigma_{DB}}. \nonumber
\end{align}
The necessary and sufficient conditions for $p>0$ are $\sqrt{b^2 - 4ac} > 0$ and $\frac{-b-\sqrt{b^2-4ac}}{2a}>0$. The first condition leads to $b^2 + 4 tr(\mC' \bm{\Sigma}_{DD}^{-1} \mC)(\Sigma_{BB}- \bm{\Sigma_{BD}} \bm{\Sigma_{DD}}^{-1} \bm{\Sigma_{DB}})>0$ and the second to $tr(\mC' \bm{\Sigma_{DD}}^{-1} \mC)(\Sigma_{BB} - \bm{\Sigma_{BD}} \bm{\Sigma_{DD}}^{-1} \bm{\Sigma_{DB}})>0$ Since $\bm{\Sigma_{DD}}$ is positive definite (thus $tr(\mC' \bm{\Sigma_{DD}}^{-1} \mC)>0$) both conditions are satisfied if $\Sigma_{BB} - \bm{\Sigma_{BD}} \bm{\Sigma_{DD}}^{-1} \bm{\Sigma_{DB}}>0$.

Given also that $A(\theta) = \rho \geq 0$ and $\bm{C(\theta)}\geq \bm{0}$, we obtain $\bm{K(\theta)} \neq 0$ and thus Assumption 5-NS is satisfied.

Now Proposition 1-NS of Komunjer and Ng (2011) can be applied, which implies that two vectors $$\vtheta_0=[\rho \ \sigma_{\tilde{y},0}^2 \ \sigma_{\nu1,0}^2 \ \sigma_{\nu2,0}^2 \ \sigma_{\zeta1,0}^2 \ \sigma_{\zeta2,0}^2]'$$ and $$\vtheta_1=[\rho \ \sigma_{\tilde{y},1}^2 \ \sigma_{\nu1,1}^2 \ \sigma_{\nu2,1}^2 \ \sigma_{\zeta1,1}^2 \ \sigma_{\zeta2,1}^2]'$$ are observationally equivalent iff there exists a scalar $\tau \neq 0$ such that
\begin{align}
A(\theta_1) & = \tau A(\theta_0) \tau^{-1} \label{eq:Atheta}\\
\bm{K(\theta_1)} & = \tau \bm{K(\theta_0)} \label{eq:Ktheta1} \\
\bm{C(\theta_1)} & = \bm{C(\theta_0)} \tau^{-1} \label{eq:Ctheta} \\
\bm{\Sigma_a(\theta_1)} & = \bm{\Sigma_a(\theta_0)}. \label{eq:Sigtheta1}
\end{align}
Given that $A(\theta)=\rho$, it follows from Equation (\ref{eq:Atheta}) that $\rho_0 = \rho_1$ and thus we can deduce from (\ref{eq:Ctheta}) that $\gamma = 1$. Hence, by using Equations (\ref{eq:Ktheta}) and (\ref{eq:Sigtheta}), the conditions (\ref{eq:Ktheta1}) and (\ref{eq:Sigtheta1}) can be expressed as
\begin{align}
\mK_1 & = \mK_0  = (p_0 \rho \mC' + \bm{\Sigma_{BD0}})(p \mC \mC' + \bm{\Sigma_{DD0}})^{-1} \label{eq:Ktheta2} \\
\bm{\Sigma_{a1}} & = \bm{\Sigma_{a0}} =p_0 \mC \mC' + \bm{\Sigma_{DD0}}, \label{eq:Sigtheta2}
\end{align}
where $p_0$ solves the following Riccati equation
\begin{equation}
p_0 = p_0 \rho^2 + \Sigma_{BB0} - \bm{K_0} (p_0 \rho \mC + \bm{\Sigma_{DB0}}). \label{eq:Ric}
\end{equation}
Equations (\ref{eq:Ktheta2}) to (\ref{eq:Ric}) are satisfied if and only if
\begin{align}
p_1 (1-\rho^2) - \Sigma_{BB1} & = p_0 (1-\rho^2) - \Sigma_{BB0}\label{eq:SigBB1}\\
p_1 \rho \mC' + \bm{\Sigma_{BD1}} & = p_0 \rho \mC' + \bm{\Sigma_{BD0}}\label{eq:SigBD1}\\
p_1 \mC \mC' + \bm{\Sigma_{DD1}} & = p_0 \mC \mC' + \bm{\Sigma_{DD0}}. \label{eq:SigDD1}
\end{align}
Without loss of generality let
\begin{equation}\label{eq:sigdelta1}
\Sigma_{\tilde{y}\tilde{y},1} =\Sigma_{\tilde{y}\tilde{y},0} + \delta(1-\rho^2)
\end{equation}
leading to
\begin{equation}\label{eq:sigdelta2}
\Sigma_{BB,1} =\Sigma_{BB,0}+ \delta(1-\rho^2).
\end{equation}

We now proceed by splitting the analysis into two cases.

{\bfseries Case 1:} $\delta = 0$. From (\ref{eq:SigBB1}) we obtain $p_1 = p_0$. (\ref{eq:SigBD1}) hence implies $\sigma_{\tilde{y,1}}^2=\sigma_{\tilde{y,0}}^2$ and $\sigma_{\nu1,1}^2 =\sigma_{\nu1,0}^2$ and given that $\Sigma_{\tilde{y}\tilde{y},1}=\Sigma_{\tilde{y}\tilde{y},0}$ it follows $\sigma_{\nu2,1}^2 =\sigma_{\nu2,0}^2$. (\ref{eq:SigDD1}) implies that  $\bm{\Sigma_{DD1}} = \bm{\Sigma_{DD0}}$ and thus $\sigma_{\zeta1,1}^2 =\sigma_{\zeta1,0}^2$ and $\sigma_{\zeta2,1}^2 =\sigma_{\zeta2,0}^2$, leading to the fact that $\theta_1 = \theta_0$.

{\bfseries Case 2:} $\delta \neq 0$. From (\ref{eq:SigBB1}) we obtain $p_1 = p_0 + \delta$. From (\ref{eq:SigBD1}) it follows
\begin{align}
\sigma_{\tilde{y},1}^2 &= \sigma_{\tilde{y},0}^2 -\delta \rho^2 \ \ \text{and} \ \
\sigma_{\nu1,1}^2 = \sigma_{\nu1,0}^2. \label{eq:solve1}
\end{align}
Moreover, (\ref{eq:SigBB1}) gives
\begin{equation}\label{eq:solve2}
\sigma_{\nu2,1}^2 = \sigma_{\nu2,0}^2 + \delta.
\end{equation}
Finally, the equations in (\ref{eq:SigDD1}) lead to
\begin{align}
\sigma_{\zeta1,1}^2&=\sigma_{\zeta1,0}^2 \ \ \text{and} \ \
\sigma_{\zeta2,1}^2 = \sigma_{\zeta2,0}^2. \label{eq:solve3}
\end{align}
Note that (\ref{eq:Sigdef}) and (\ref{eq:solve1}) to (\ref{eq:solve3}) result into
 \begin{equation}\label{def:Sigmasolve}
\mSigma_1 =  \begin{bmatrix}
\Sigma_{\tilde{y}\tilde{y},0} + \delta(1-\rho^2) & \Sigma_{\tilde{y}\nu1,0} - \delta & \Sigma_{\tilde{y}\nu2,0} - \delta & 0 & 0 \\
\Sigma_{\nu1\tilde{y},0} - \delta & \Sigma_{\nu1\nu1,0} + \delta & \Sigma_{\nu1\nu2,0} + \delta & 0 & 0 \\
\Sigma_{\nu2\tilde{y},0} -\delta & \Sigma_{\nu2\nu1,0} + \delta & \Sigma_{\nu2\nu2,0} + \delta & 0 & 0 \\
0 & 0 & 0 & \Sigma_{\zeta1\zeta1} & 0 \\
0 & 0 & 0 & 0 & \Sigma_{\zeta2\zeta2} \\
\end{bmatrix}.
\end{equation}
Finally, from (\ref{eq:solve2}) and (\ref{eq:solve3}) it follows that $\delta=0$.
\end{proof}

\section*{Identification in generalized reconciliation models}

We now return to the use of moment conditions to discuss possible approaches to identification in data reconciliation models with multiple data releases. We only consider the reconciliation of exactly two data series, but otherwise consider linear dynamic data generating processes more general than any that we have seen used in the literature. 

Specifically, we consider models of the form 
\begin{eqnarray}
    \begin{bmatrix}
        \mY_{t}^{0} \\
        \mY_{t}^{1}%
    \end{bmatrix}
    = \mZ \cdot \valpha_{t}
\end{eqnarray}%
and 
\begin{equation}
    \valpha_{t}  = \mT \cdot \valpha_{t-1} + \mR \cdot \veta_{t}
\end{equation}
where
\begin{itemize}
\item $\mY_{t}^{i}$ is a $1\times l$ vector containing $l$
releases of series $i=\{0,1\}~$estimates of the unobserved true value $%
\widetilde{y}_{t}.$

\item $\valpha _{t}$ is a $\left( p+4\cdot l\right) \times 1$ latent state
vector that we may partition 
as $\valpha _{t}^{\prime }=%
\begin{bmatrix}
\widetilde{\mY}_{t-1}^{p\prime } & \vnu_{t}^{0\prime } &
\vnu _{t}^{1\prime } & \vzeta_{t}^{0\prime }
& \vzeta_{t}^{1\prime }%
\end{bmatrix}%
$

\item $\widetilde{\mY}_{t-1}^{p}$ is a $p\times 1$ vector $\equiv \left[
\widetilde{y}_{t},\ldots ,\widetilde{y}_{t-p+1}\right] ^{\prime }$

\item $\vnu_{t}^{i}$ for $i=\{0,1\}$ is a $l\times 1$
vector of news shocks contained in $\mY_{t}^{i}$

\item $\vzeta_{t}^{i}$ for $i=\{0,1\}$ is a $l\times 1$
vector of noise shocks contained in $\mY_{t}^{i}$

\item $\veta _{t}\sim N\left( \bm{0},\mSigma \right) $ is a $\left( 1+4\cdot
l\right) \times 1$ vector of i.i.d. mean zero normally distributed
shocks.with diagonal covariance matrix $\mSigma $.

\item $\mZ$ is a $\left( 2\cdot l\right) \times \left( p+4\cdot l\right) $
matrix of the form%
\[
Z\equiv
\begin{bmatrix}
\bm{1}_{\left( 2\cdot l\right) \times 1} & \bm{0}_{\left( 2\cdot l\right) \times
\left( p-1\right) } & \mI_{2\cdot l} & \mI_{2\cdot l}%
\end{bmatrix}%
\]

\item $\bm{1}_{a\times b}$ is a matrix of dimension $a\times b$ composed entirely
of $1$'s

\item $\bm{0}_{a\times b}$ is a matrix of dimension $a\times b$ composed entirely
of $0$'s

\item $\mI_{a}$ is a $a\times a$ identity matrix

\item $\mT$ is a $\left( p+4\cdot l\right) \times \left( p+4\cdot l\right) $
block diagonal matrix of the form
\begin{equation}
\mT\equiv
\begin{bmatrix}
\mT_{p} & \bm{0}_{p\times \left( 4\cdot l\right) } \\
\bm{0}_{\left( 4\cdot l\right) \times p} & \mT_{S}%
\end{bmatrix}%
\end{equation}

\item $\mT_{p}$ is a $p\times p$ matrix
\begin{equation}
\mT_{p}=%
\begin{bmatrix}
\rho _{1},\rho _{2},\ldots  & \rho _{p} \\
\mI_{p-1} & \bm{0}_{\left( p-1\right) \times 1}%
\end{bmatrix}%
\end{equation}

\item $\mT_{S}$ is a $\left( 4\cdot l\right) \times \left( 4\cdot l\right) $
arbitrary diagonal matrix

\item $\mR$ is a $\left( p+4\cdot l\right) \times \left( 1+4\cdot l\right) $
matrix of the form%
\begin{equation}
\mR=%
\begin{bmatrix}
1 & \bm{1}_{1\times l} & \bm{1}_{1\times l} & \bm{0}_{1\times l} & \bm{0}_{1\times l} \\
\bm{0}_{\left( p-1\right) \times 1} & \bm{0}_{\left( p-1\right) \times l} & \bm{0}_{\left(
p-1\right) \times l} & \bm{0}_{\left( p-1\right) \times l} & \bm{0}_{\left( p-1\right)
\times l} \\
\bm{0}_{l\times 1} & -\mU_{l} & \mPsi  & \bm{0}_{l\times l} & \bm{0}_{l\times l} \\
\bm{0}_{l\times 1} & \bm{0}_{l\times l} & -\mU_{l} & \bm{0}_{l\times l} & \bm{0}_{l\times l} \\
\bm{0}_{l\times 1} & \bm{0}_{l\times l} & \bm{0}_{l\times l} & \mI_{l\times l} & \mPhi  \\
\bm{0}_{l\times 1} & \bm{0}_{l\times l} & \bm{0}_{l\times l} & \bm{0}_{l\times l} & \mI_{l\times l}%
\end{bmatrix}%
\end{equation}

\item $\mU_{l}$ is a $l\times l$ matrix with $0$'s below the main diagonal and
$1$'s everywhere else

\item $\mPsi ,\mPhi $ are unrestricted $l\times l$ matrices.
\end{itemize}

The model estimated in the paper is the special case of the above where
\begin{enumerate}
    \item $p = 1$
    \item $\mT_S = \bm{0}_{(4 \cdot l) \times (4 \cdot l)}$
    \item $\mPhi = \mPsi = \bm0_{l \times l} $
\end{enumerate}
Relaxing the first condition allows us to consider model where the dynamics of the unobserved true values follow an $AR(p)$ process rather than simply an $AR(1)$. Allowing for higher-order autocorrelations adds an additional $p-1$ free parameters to the model, but also adds an additional $2 \cdot (p-1)$ sample autocorrelations that may be used for identification.

Relaxing the second condition allows what JvN refer to as ``spillover'' effects. This permits revisions to the values for calendar period $t$ to be correlated with revisions to calendar period $t-1$. This may occur, for example, when revisions tend to shift measured growth from one quarter to an adjascent quarter, or when the incorporation of lower frequency data sources (e.g. annual tax returns) shift multiple periods in the same direction. This adds an additional $4 \cdot l$ free parameters to the model. However, it also brings into play an additional $2 \cdot (l - 1)$ moments capturing the 1st-order autocorrelations of the revisions of our two series.

Relaxing the third condition allows for the possibility that measurement errors of either type may be correlated across the two series. Contemporaneous correlations (i.e. measurement errors that affect the same \textit{release} of each series) are captured by the diagonals of these two matrices. Evidence that information tends to be incorporated into releases of $y^0$ before (after) those of $y^1$ implies that there should be non-zero entries of $\mPsi$ above (below) the main diagonal. Contemporaneous correlations would add an additional $2 \cdot l$ free parameters to the model, while unrestricted correlations would add an additional $2 \cdot l^2$ free parameters. However, we have already assumed the use of all $l \cdot (2 \cdot l + 1)$ contemporaneous cross-moments of the various vintages of both series, so there is no offsetting gain in the number of moments available for identification.

\clearpage

\printbibliography 
\end{document}